\documentclass{article}

\usepackage[english]{babel}

\usepackage[letterpaper,top=2cm,bottom=2cm,left=3cm,right=3cm,marginparwidth=1.75cm]{geometry}

\usepackage{amsmath}
\usepackage{graphicx}
\usepackage[colorlinks=true, allcolors=blue]{hyperref}

\usepackage{amsmath,amsfonts}
\usepackage{amsthm}
\usepackage{amssymb}
\usepackage{enumitem}
\usepackage{floatrow}
\usepackage{csquotes}

\newtheorem{claim}{Claim}
\newtheorem{corollary}{Corollary}
\newtheorem{defn}{Definition}

\usepackage{graphicx}
\usepackage[caption=false]{subfig}

\usepackage{biblatex} 
\newcommand{\Addresses}{{
  \bigskip
  \footnotesize

  T.~Gilat, \textsc{Department of Computer Science, Bar-Ilan University,
    Ramat Gan, Israel 5290002.}\par\nopagebreak
  \textit{E-mail:} \texttt{tom.gilat@biu.ac.il}

\medskip

  B.~Gilat, \textsc{Department of Computer Science, Tel Aviv University, Tel Aviv, Israel 6997801.}\par\nopagebreak
  \textit{E-mail:} \texttt{ben.gilat@mail.tau.ac.il}

}}

\begin{document}

\title {Fast Conformal Parameterization\\ of Disks and Sphere Sectors}

\author{Tom Gilat and Ben Gilat}


\date{\today}
\maketitle

\begin{abstract}

We prove a novel method for the embedding of a 3-fold rotationally symmetric sphere-type mesh onto a subset of the plane with 3-fold rotational symmetry. The embedding is free-boundary with the only additional constraint on the image set is that its translations tile the plane, in turn this forces the angles at the embedding of the branch points in the construction. These parameterizations are optimal with respect to the Dirichlet energy functional defined on simplicial complexes. Since the parameterization is over a fixed area domain, it is conformal (i.e. a minimizer of the LSCM energy). The embedding is done by a novel construction of a torus from 63 copies of the original sphere. As a foundation for this result we first prove the optimality of the embedding of disk-type meshes onto special types of triangles in the plane, and rectangles. The embedding of the 3-fold symmetric torus is full rank and so cannot be reduced by simpler constructions. 3-fold symmetric surfaces appear in nature, for example the surface of the 3-fold symmetric proteins PIEZO1 and PIEZO2 which are an important target of current studies. 
\end{abstract}

\maketitle

\section{Introduction}

 In this work we start by describing a method for parametrizing a simplicial complex, which is a combinatorial disk, embedded in $\mathbb{R}^3$ onto three types of domains in $\mathbb{R}^2$. By "parametrizing" we mean that we give an efficient way to compute its embedding. The three types of domains are a right-angled isosceles triangle, an equilateral triangle and a square. These parameterizations are optimal with respect to the Dirichlet energy functional defined on simplicial complexes. In practice one can write a suitable system of linear equations for the coordinates of the images in the plane of the vertices of the embedding. This method and its related proofs were described in \cite{OrbifoldTutte}. However, by giving direct proofs which utilize in a novel way the symmetry of the constructed simplicial complex, we are able to address and prove the main contribution of the article, namely, parameterizing a 3-fold rotationally symmetric sphere-type mesh onto a tile in the plane, which has 3-fold rotational symmetry and a free-boundary under the constraint that the plane can be covered with translations of such a tile. This parameterization is optimal and it is rendered by solving two full rank linear systems. 
 
 \begin{figure}[ht!]
\begin{center}
\hspace*{-0.4in}
\includegraphics[trim=0cm  0cm 0cm 0cm ,clip, scale=0.64]{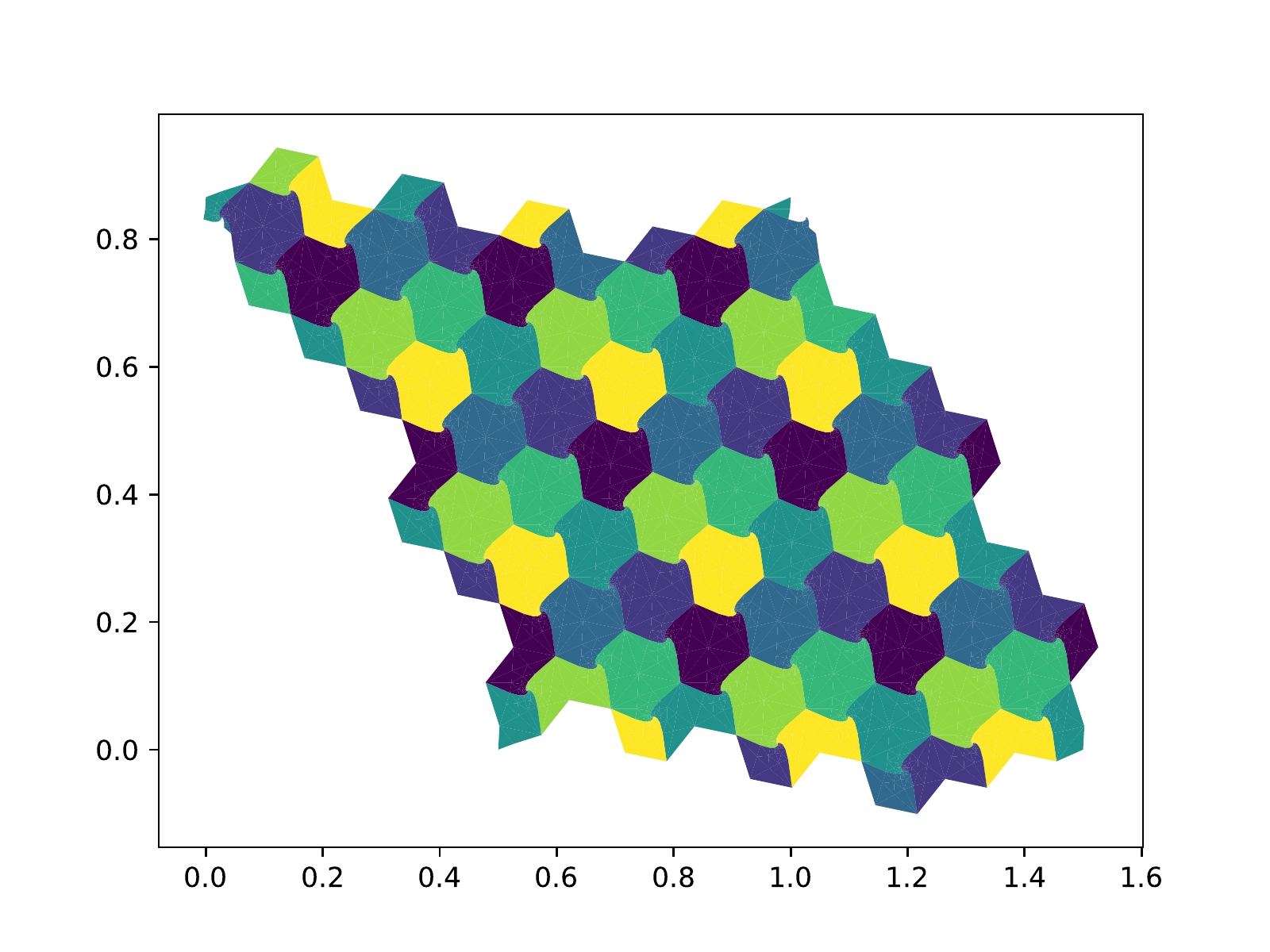}
\caption{Embedding of the torus constructed out of 63 sphere copies (using equal weights for simplicity). }
\label{fig:embedding_noedge_full}
\end{center}
\end{figure}
 
 For the matter of the proof, in each one of the cases we construct out of copies of the given mesh, a suitable simplicial complex embedded in $\mathbb{R}^3$ which is homeomorphic to a torus. We then embed this torus in a suitable plane torus, i.e. a quotient space of $\mathbb{R}^2$ by a rank 2 lattice, with weights which are calculated according to the geometry of the given surface and guarantee minimization of the energy considered. The validity of the embedding is due to proofs by Lovasz \cite[Chapter 7, p.~ 91-118]{lovász2019graphs} and Gortler, Gotsman and Thurston \cite{GortlerTorus}. 
 
 Parameterizations which minimize the Dirichlet energy where the mapping has a target domain of fixed area, minimize the LSCM energy which correspond to the notion of conformal mappings (see \cite{LSCMEnergy}). Conformal mappings are highly desired in the area of Geometry Processing, see for example \cite{Dym732}. The equivalent notion of discrete holomorphic 1-forms first appeared in \cite{gu2003global} and \cite{jin2004optimal}. For the theory about convergence of discrete holomorphic 1-forms to conformal functions on surfaces see  \cite{acta_numericaFEM}. Discrete conformal functions also have a formalism via circle packings, see \cite{stephenson2005introduction}. For a very readable treatise on holomorphic 1-forms (or discrete analytic functions), see \cite{Analytic,lovász2019graphs}.
 
 This method for embedding a 3-fold rotationally symmetric mesh has valuable applications. It has been shown that symmetry in protein structures has important roles in robustness to perturbation and in function \cite{symmetryPNAS}. The 2021 Nobel Prize for Physiology and Medicine was awarded to researchers who studied the PIEZO1 and PIEZO2 proteins which have 3-fold rotationally symmetric structure. 3-fold symmetric images also appear in different symbolism. See Figure~\ref{fig:piezo2}. Using the presented method one can rapidly obtain a minimally distorted 2d image of the interface surface of such proteins which would enable machine learning studies concerning the effect of mutations at the surface and the surface characteristics. Note that a 3-fold rotationally symmetric surface of a structure, can be made from 3 copies of a disk-type surface, with a connecting scaffold, and can then be used in similar studies.
 
 The proofs use a simple symmetry based framework consisting of identifying a symmetry in the image domain, the plane torus in this case, and looking for a corresponding symmetric surface which can be constructed from copies of the given surface. For the different symmetries we define conjugate maps, one in the image domain and one in the source domain. These maps, when composed on the right side and on the left side respectively with the optimal embedding map, yield the same optimal embedding map. We then conclude that the embedding has a corresponding symmetry. Applying this type of argument with different pairs of conjugate maps prove the desired symmetries of the embedding. Then we conclude optimality of the embedding map restricted to each one of the copies consisting of the constructed surface. 
 
 The most elaborate application of this method, and the most useful in our opinion, is the construction of a torus out of 63 copies of a 3-fold rotationally symmetric sphere-type mesh. This number of copies is the minimal number that allows the construction of a branched covering with the same ramification structure for each of the branching points positioned in identical positions in 3d space. This can be validated by examining the Riemann-Hurwitz formula which relates the ramification structure of branched covering maps to the genus of the domain (see \cite{Haim2019SurfaceNV} for details), and taking into consideration the geometric constraints required for the stitching of the copies. We are able to parametrize the torus onto a 120 degrees rhombus as shown in Figure \ref{fig:embedding_noedge_full}, by solving a full rank system of linear equations, whose size is approximately 63 times the number of vertices in the original sphere-type mesh. It is interesting to compare this to the embedding of a sphere described in  \cite{OrbifoldTutte} onto "Orbifold of type 2". If our symmetric sphere was constructed out of 3 spheres, we would expect to obtain the same result as they did. However, for any other 3-fold rotationally symmetric sphere, our method which is specialized for 3-fold symmetric meshes and embeds the sphere using a much larger full rank linear system will yield minimally distorted parameterizations. We refer to the three symmetric parts of the sphere-type mesh as sphere sectors. 
 
 The paper starts with the case of embedding a complex onto a right-angled isosceles triangle. In Section 2 we construct the torus to be embedded for that case. In Section 3 we describe the energy functional to be minimized and explain what the weights are. In Section 4 we prove the desired symmetries of the embedding and conclude optimality of the embedding of a single copy of the given complex. In Section 5 we deal rigorously with the case of embedding onto an equilateral triangle. The case of embedding onto a rectangle is described very briefly in Section 6. In Section 7 we remark about the complexity of implementing the embedding. Then we deal with the most elaborate case - embedding of a sphere with 3-fold rotational symmetry.



\begin{figure}%

    \subfloat[\centering The flag of the Isle of Man]{\includegraphics[trim = 0cm  18cm 5cm 1cm ,clip, scale=0.45]{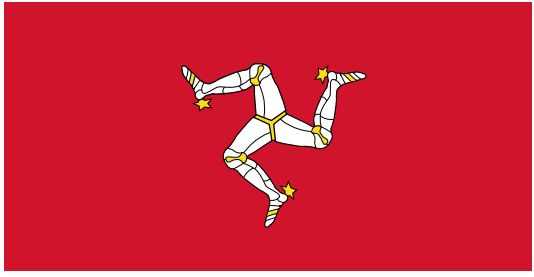}}%
    \qquad\qquad
   \subfloat[\centering Li Wang et al., Cryo-EM Structure of the Mammalian Tactile Channel Piezo2, 2019  (doi:10.2210/pdb6KG7/pdb) \cite{Wang2019}]{{\includegraphics[trim=5cm  8cm 5cm 9.5cm ,clip, scale=0.40]{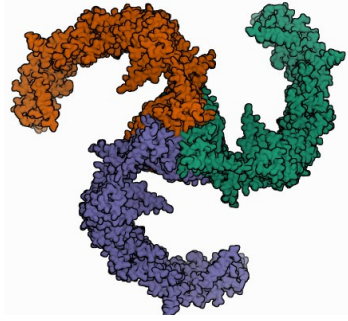}}}
    \caption{3-fold symmetry in symbolism and in nature.}
    \label{fig:piezo2}
\end{figure}



\section{Constructing the torus}
We begin with the case of embedding a complex onto a right-angled isosceles triangle.

We start with two definitions (see also \cite{stephenson2005introduction}).
\begin{defn}
A surface is a connected topological 2-manifold, that is, a connected Hausdorff space $S$ in which each point has a neighborhood that is homeomorphic to an open subset of the plane. We require the transition maps to be continuous. 
\end{defn}

\begin{defn}
A disc-type mesh is is a combinatorial disc (that is, simplicial 2-complex which is simply connected, finite and with nonempty boundary) embedded in $\mathbb{R}^3$ (i.e. it is a surface, 1-simplices are embedded as straight lines).
\end{defn}

We construct $\widetilde{\mathcal{M}}$, a simplicial 2-complex immeresed in $\mathbb{R}^3$, with torus topology. By torus topology, we mean that it is homeomorphic to $\mathbb{R}^2/\mathbb{Z}^2$.


Let $\mathcal{M}$ be a disc-type mesh. By definition $\mathcal{M}$ has a nonempty boundary and therefore its boundary contains at least 3 distinct points. We choose three distinct veritices on the boudary of $\mathcal{M}$: $v_0,v_1,v_2\in\mathbb{R}^3$. We let $\Gamma_0$ be the path between $v_0$ and $v_1$ on the boundary of the disc. Let $\Gamma_1$ be the path on the disc's boundary between $v_1$ and $v_2$, and let $\Gamma_2$ be the path on the disc's boundary between $v_2$ and $v_0$. Simply connectedness implies that these paths are uniquely defined. 

We now define eight duplicates of  the disk-type mesh $\mathcal{M}$, this means eight identical combinatorial disks embedded identically in $\mathbb{R}^3$. Denote them by $\mathcal{M}_1,\mathcal{M}_2,...\mathcal{M}_8$.
For each copy $\mathcal{M}_i$ $(i=1,...,8)$, let $v_0^i,v_1^i,v_2^i$ be the vertices positioned at $v_0,v_1,v_2$ respectively. For each copy $\mathcal{M}_i$ $(i=1,...,8)$, let $\Gamma_0^i$ be the path between $v_0^i$ and $v_1^i$ on the boundary of $\mathcal{M}_i$. Let $\Gamma_1^i$ be the path on the $\mathcal{M}_i$'s boundary between $v_1^i$ and $v_2^i$, and let $\Gamma_2^i$ be the path on $\mathcal{M}_i$'s boundary between $v_2^i$ and $v_0^i$.

We now construct a simplical 2-complex $\widetilde{\mathcal{M}}_\Delta$, immersed in $\mathbb{R}^3$, out of the copies $\mathcal{M}_i, i=1,...,8$. This is done by uniting the following pairs of paths: $\Gamma_0^{i}$ is united with $\Gamma_0^{i+1}$ for $i=1,3,5,7$. We name the united paths $\widetilde{\Gamma}_0^1,\widetilde{\Gamma}_0^2, \widetilde{\Gamma}_0^3, \widetilde{\Gamma}_0^4$ respectively. Secondly, we unite the following pairs: $\Gamma_1^1$ with $\Gamma_1^8$, and $\Gamma_1^i$ with $\Gamma_1^{i+1}$ for $i=2,4,6$. We name these united paths $\widetilde{\Gamma}_1^1,\widetilde{\Gamma}_1^2, \widetilde{\Gamma}_1^3, \widetilde{\Gamma}_1^4$ respectively. Lastly, we unite the following pairs: $\Gamma_2^1$ with $\Gamma_2^4$, 
$\Gamma_2^2$ with $\Gamma_2^7$, $\Gamma_2^3$ with $\Gamma_2^6$, and $\Gamma_2^5$ with $\Gamma_2^8$. We name these united paths $\widetilde{\Gamma}_2^1,\widetilde{\Gamma}_2^2, \widetilde{\Gamma}_2^3, \widetilde{\Gamma}_2^4$ respectively.

We will show that $\widetilde{\mathcal{M}}_\Delta$ is homeomorphic to the 2-torus. We will refer to the 8 different sub-simplices of $\widetilde{\mathcal{M}}_\Delta$ corresponding to $\mathcal{M}_i, i=1,...,8$, as $\mathcal{M'}_i, i=1,...,8$ in accordance. Note that the $\mathcal{M}'_i$'s are not disjoint, as we united edges and vertices to construct $\widetilde{\mathcal{M}}_\Delta$.
We use Tutte embedding (see \cite[Chapter 3, p.~ 23, Thm.~3.2]{lovász2019graphs}), to map each $\mathcal{M'}_i$ to $\mathrm{M}'_i$ in the $\mathbb{R}^2/\mathbb{Z}^2$ according to the following diagram. Each $\mathcal{M}'_i, i=1,...,8$, is mapped to $\mathrm{M}'_i$ in the diagram, where each bordering path $\widetilde{\Gamma}_j^k, j\in \{0,1,2\}, k\in \{1,2,3,4\}$ is mapped onto $\gamma_j^k$ up to a translation.
We demand that the vertices on each path are mapped to equally spaced vertices on the line it is mapped onto, and of course that the extreme vertices of each path are mapped to the extreme points of that line. If we consider the well defined total map from $\widetilde{\mathcal{M}}_\Delta$ to $\mathbb{R}^2/\mathbb{Z}^2$, we have a homeomorphism based on the properties of Tutte embedding.

From now on, we refer to $\widetilde{\mathcal{M}}_\Delta$ as $\widetilde{\mathcal{M}}$. We regard $\widetilde{M}$ as a surface, but also regard the immersed (or embedded) subsets of the underlying simplicial complex structure of $\widetilde{M}$. 

\begin{figure}
\begin{center}
\qquad\qquad
\includegraphics[trim=4.5cm  5cm 5cm 3cm ,clip, scale=0.7]{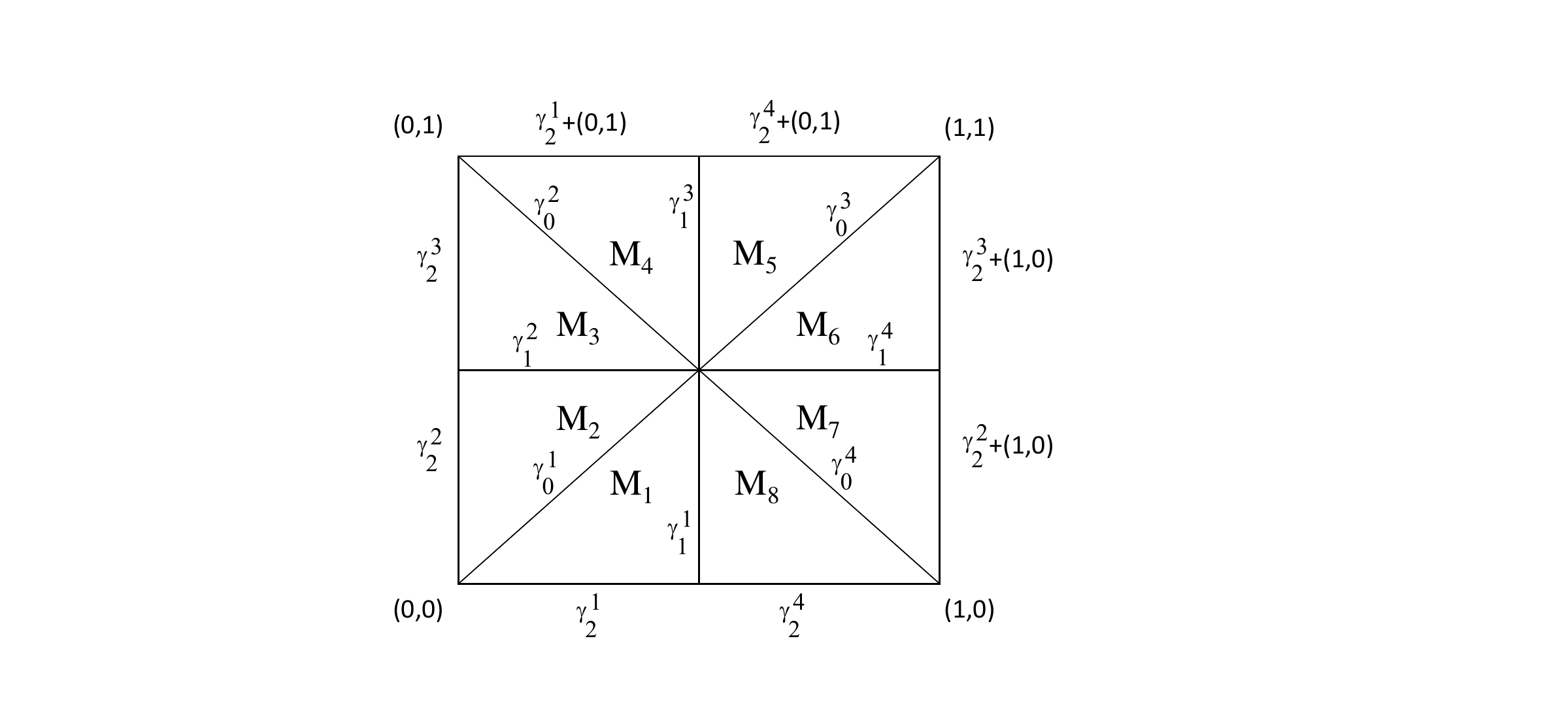}
\caption{Embedding of the 8 disks construction onto $\mathbb{T}^2=\mathbb{R}^2/\mathbb{Z}^2$.}
\label{fig:torus1}
\end{center}
\end{figure}

\section{Optimal embedding of the torus}
Given two vectors $v_1,v_2\in\mathbb{R}^2$, let $\Lambda$ be the lattice generated by $v_1,v_2$. We wish to find an embedding of $\widetilde{\mathcal{M}}$ onto $\mathbb{R}^2/\Lambda$, which minimizes the Dirichlet energy functional. The proof for the validity of the embedding itself is due to results by Gortler, Gotsman and Thurston or Lovasz \cite{GortlerTorus}\cite{lovász2019graphs}. 

In the general setting the Dirichlet energy of a map $U$ is defined as follows. Let $\mathcal{X}=\mathcal{X}^n$ and $\mathcal{Y}=\mathcal{Y}^m$ be two smooth compact Riemannian manifolds of dimension $n$ and $m$, respectively. We assume $\mathcal{X}$ and $\mathcal{Y}$ are equipped with metric tensors $(g_{\alpha\beta})$ and $(\gamma_{ij})$, respectively, in some local coordinate charts $(x_1,...,x_n)$ at $x$, and $(U^1,...,U^m)$ at $U(x)$ on $\mathcal{X}$ and $\mathcal{Y}$, respectively. The \textit{Dirichlet energy} of a smooth map $U:\mathcal{X}\rightarrow\mathcal{Y}$ is defined as the integral of the square of the derivative $dU$. More precisely, the \textit{energy density} of $U$ is
\[
    e(x,U):=\frac{1}{2}|dU_x|^2=\frac{1}{2}\mathrm{tr}[(dU_x)^* dU_x]
\]
(the Hilbert-Schmidt norm is independent of the choice of frame) and the Dirichlet energy of $U$ is 
\[
    \mathcal{E}(U,\mathcal{X}):=\int_\mathcal{X} e(x,U)d\mathrm{vol}_\mathcal{X}=\frac{1}{2}\int_\mathcal{X}|dU_x|^2 d\mathrm{vol}_\mathcal{X}\ .
\]
And so,
\[
      \mathcal{E}(U,\mathcal{X})=\frac{1}{2}\int_{\mathcal{X}} g^{\alpha\beta}\gamma_{ij}\frac{\partial U^i}{\partial x^\alpha}\frac{\partial U^j}{\partial x^\beta} d\mathrm{vol}_\mathcal{X}\ .
\]
The critical points of the Dirichlet energy functional are the harmonic functions, see \cite{EellsSampson}.

According to Theorem 3.3 in \cite{GortlerTorus}, we have a function $\widetilde{U}$ from the set of all faces, $\widetilde{\mathcal{M}}$, onto $\mathbb{R}^2/\Lambda$. For $\Delta\in\widetilde{\mathcal{M}}$ and $\widetilde{U}(\Delta)$ the metric tensors are the ones induced from the Euclidean Riemannian metric on $\mathbb{R}^3$. The simplicial Dirichlet energy is then defined to be
\[
    \mathcal{E}(\widetilde{U},\widetilde{\mathcal{M}}):=\sum_{\Delta\in\widetilde{\mathcal{M}}}\mathcal{E}(\widetilde{U}|_\Delta,\Delta).
\]

In applications one should use Pinkall and Polthier classic work \cite{pinkall1993} for the calculation of the "cotangent weights" which yield the optimal embedding, $\widetilde{U}$, with respect to the simplicial Dirichlet energy functional. Note that in the embeddings we consider, we either let the boundary be free (in the 3-fold symmetric case) or constrain the points to a line on the boundary (in the case of the triangles and the square) but we do not fix them and thus cannot use Pinkall and Polthier classic work as is since it requires fixing the points on the boundary. The image of each vertex in $\widetilde{\mathcal{M}}$ is mapped to a weighted average of the images of its neighbours according to the following equation
\begin{equation*}
    \sum_{i=1}^n w_i (x_i-x_0)\equiv\sum_{i=1}^n\left(\cot \alpha_i + \cot \beta_{i}\right)(x_i-x_0)=0,\quad x_i\in\mathbb{R}^2,
\end{equation*}
where $x_i, i=1,...,n$, is an ordering of the 1-ring of $x_0$, and $(\alpha_i)_{i=1}^n,(\beta_i)_{i=1}^{n}$ are defined as follows. Set $x_{n+1}=x_1$. For $i=1,...,n$, let $\alpha_i$ be the angle between the lines $\overline{x_0 x_i}$ and $\overline{x_i x_{i+1}}$, and let $\beta_i$ be the angle between the lines $\overline{x_0 x_{i+1}}$ and $\overline{x_i x_{i+1}}$. In the  calculation of the weights, the ordering of the 1-rings should be in the direction which keeps a consistent orientation of the faces of the simplicial complex. For simplicity we used $w_i=1$ for all $i$. As mentioned, in our set of equations the boundaries were not fixed. Instead we fixed two vectors that generate a lattice $\Lambda$, and we wrote the system of equations in $\mathbb{R}^2/\Lambda$. Compare with the equations for mapping the torus that can be found in [p.~101, Eq. 5,6]\cite{GortlerTorus}, which imply the equations we used.

Throughout this work we will assume that the weights are positive. This is the case if no obtuse angles occur, or more specifically one uses intrinsic Delaunay triangulation for the surface.

\section{Properties of the mapped torus}
We now show the symmetry properties of the one-to-one Dirichlet energy minimizing mappings of $\widetilde{\mathcal{M}}$ onto $\mathbb{R}^2/\Lambda$, where $\Lambda$ is a lattice generated by two unit vectors, i.e. $\Lambda = \langle v_1,v_2\rangle$ and $\lVert v_1 \rVert=\lVert v_2 \rVert=1$. (Norm is the 2-norm.)

We first examine the one-to-one Dirichlet energy minimizing mapping of $\widetilde{\mathcal{M}}$ onto $\mathbb{R}^2/\widetilde{\Lambda}$, where $\widetilde{\Lambda} = \langle e_1,e_2\rangle,\ e_1=(1\  0)^T, e_2=(0\ 1)^T$, i.e. the case where $\Lambda=\mathbb{Z}^2$.

Let $l_H = \mathbb{R}\begin{pmatrix}1 \\ 0\end{pmatrix}$ and  $l_V= \mathbb{R}\begin{pmatrix}0 \\ 1\end{pmatrix}$, two infinite lines in the plane. Let $R_H,R_V$ be the reflections along $l_H,l_V$ respectively. In addition, let $l_{DP}=\mathbb{R}\begin{pmatrix}1 \\ 1\end{pmatrix}$ and $l_{DS}=\mathbb{R}\begin{pmatrix}1 \\ -1\end{pmatrix}+\begin{pmatrix}0 \\ 1\end{pmatrix}$, corresponding to the primary and second diagonal of the square with corners $(0\ 0)^T, (1\ 1)^T$, a fundamental domain in $\mathbb{R}^2$ for $\mathbb{R}^2/\widetilde{\Lambda}$. Let $R_{DP}, R_{DS}$ be reflections along $l_{DP},l_{DS}$ respectively.

\begin{claim}
The one-to-one Dirichlet energy minimizing mapping of $\widetilde{\mathcal{M}}$ onto $\mathbb{R}^2/\widetilde{\Lambda}$ is symmetric with respect to reflections along $l_H,l_V,l_{DP},l_{DS}$.
\end{claim}
\begin{proof}
Let $\widetilde{X}$ be the simplicial complex, which is an embedding of $\widetilde{\mathcal{M}}$ by the one-to-one Dirichlet energy minimizing map of $\widetilde{\mathcal{M}}$ onto $\mathbb{R}^2/\widetilde{\Lambda}$. We denote this map by $\Phi$ ($\Phi:\widetilde{\mathcal{M}}\rightarrow\mathbb{R}^2/\widetilde{\Lambda}$). It carries 2-simplices to 2-simplices, 1-simplices to 1-simplices, and so on. We extend in a natural way $R_H,R_V,R_{DP},R_{DS}$ to be maps on simplicial complexes in the plane, and regard $\widetilde{X}$ as a simplicial complex in the fundamental domain. Note that the image of an embedding of a simplicial complex on the fundemantal domain under the reflection operations is still an embedding of the complex onto $\mathbb{R}^2/\widetilde{\Lambda}$.

We define the maps $S_H,S_V,S_{DP},S_{DS}$ which are all automorphisms of the surface $\widetilde{\mathcal{M}}$, in the sense that they are one-to-one, onto, and continuous. We define the map $S_H$ as the "horizontal reflection" (formal explanation in the next lines) of the surface $\widetilde{\mathcal{M}}$. $S_H$ reflects $\widetilde{\mathcal{M}}$ along $\widetilde{\Gamma}_1^1$ concatenated by $\widetilde{\Gamma}_1^3$, meaning it maps $\mathcal{M}_1$ to $\mathcal{M}_8$, $\mathcal{M}_2$ to $\mathcal{M}_7$ and so on. (In general for $i=1,2,3,4$ it maps $\mathcal{M}_i$ to $\mathcal{M}_{9-i}$.) All these maps, each a restriction of $S_H$, are the natural maps from one disk-type mesh to another copy of the same disk-type mesh. $S_H$, defined this way, is an automorphism of $\widetilde{\mathcal{M}}$. The maps $S_V, S_{DP},S_{DS}$ are defined in a similar manner. They are the "reflections" of $\widetilde{\mathcal{M}}$ along: (Case V) $\widetilde{\Gamma}_1^2$ concatenated by $\widetilde{\Gamma}_1^4$, (Case DP) $\widetilde{\Gamma}_0^1$ concatenated by  $\widetilde{\Gamma}_0^3$, (Case DS) $\widetilde{\Gamma}_0^2$ concatenated by  $\widetilde{\Gamma}_0^4$; respectively.

We have the following: $\Phi=R_H\circ\Phi\circ S_H$, but also $S_H \widetilde{\mathcal{M}}\cong\widetilde{\mathcal{M}}$ (the two sides are equal up to renaming), and so $\Phi(\widetilde{\mathcal{M}})=R_H\circ\Phi(\widetilde{\mathcal{M}})$. (The first equation can be seen to hold from observing that both sides satisfy the fully constrained system of linear equations as in Gortler, Gotsman and Thurston 's embedding of a torus in the plane.) The same argument applies for $R_V$ and $S_V$, $R_{DP}$ and $S_{DP}$, $R_{DS}$ and $S_{DS}$. This finishes the proof of the claim.




\end{proof}
\begin{corollary}
$\Phi (\mathcal{M}_1)$ is precisely one of the octants of the fundamental domain of $\mathbb{R}^2/\widetilde{\Lambda}=\mathbb{R}^2/\mathbb{Z}^2$, and is a right angle isosceles triangle with leg length $\frac{1}{2}$ (see Figure \ref{fig:torus1}).
\end{corollary}
\begin{proof}
Denote by $O$ the octant of the fundamental domain of $\mathbb{R}^2/\mathbb{Z}^2$ which has non-empty intersection with $\Phi(\mathcal{M}_1)$. Then if it does not hold that $O\subset\Phi(\mathcal{M}_1)$, by the symmetries of $\Phi(\widetilde{\mathcal{M}})$ we get a contradiction for the fact that $\Phi$ is an embedding onto $\mathbb{R}^2/\mathbb{Z}^2$. On the other hand, if it is strictly that $O\subset\Phi(\mathcal{M}_1)$, then we get a contradiction for the assumption that $\Phi$ is one-to-one on its range, and so we conclude the statement.
\end{proof}
\begin{corollary}
$\mathcal{E}(\Phi|_{\mathcal{M}_1},\mathcal{M}_1)=\min_{\phi\in\mathcal{A}}\mathcal{E}(\phi,\mathcal{M}_1)$, where $\mathcal{A}$ is the set of embeddings of $\mathcal{M}$ onto the triangle with corners $(0,0), (0.5,0.5), (0,0.5) $ such that $v_0,v_1,v_2$ are mapped to these corner points respectively.
\end{corollary}
\begin{proof}
 This is clear - as if this is not the case, there exists $\phi\in\mathcal{A}$ such that $\mathcal{E}(\Phi|_{\mathcal{M}_1},\mathcal{M}_1)>\mathcal{E}(\phi,\mathcal{M}_1)$, and we can construct a map $\Phi'$ for which $\mathcal(\Phi',\widetilde{\mathcal{M}})<\mathcal(\Phi,\widetilde{\mathcal{M}})$. Contradiction.
\end{proof}
 
\section{Mapping to an equilateral triangle}

We show how to use the torus embedding to embed a disk-type mesh onto an equilateral triangle in the plane. Again, assume that we have $v_0,v_1,v_2$, 3 marked vertices on the boundary of the mesh, $\mathcal{M}$. As shown in Figure \ref{fig:hex_tile}, the plane can be tiled by regular hexagons, and one can be convinced that exactly seven regular hexagons can be arranged on a flat torus whose fundamental domain is a rhombus with angles of 60 degrees and 120 degrees, such that their interiors are disjoint. We define $w_1=\begin{pmatrix}1 \\ 0\end{pmatrix}$, $w_2 = \begin{pmatrix}\mathrm{Re}(e^{i\frac{\pi}{3}}) \\ \mathrm{Im}(e^{i\frac{\pi}{3}}) \end{pmatrix}$, and $\Lambda_R = \langle w_1,w_2\rangle$. Each hexagon in the tiling of the plane, can be divided into six regular triangles with disjoint interiors having a common corner at the center of the hexagon. We define an equivalence relation on the triangles consisting of the tiling of the plane, resulting from a refinement of the hexagonal tiling by this division. Two triangles $A,B\subset\mathbb{R}^2$ are equivalent if exist $a_1,a_2\in\mathbb{Z}$, such that $A+a_1 w_1+a_2 w_2=B$. We are now ready to explain how to construct a torus $\widetilde{\mathcal{M}}$ from $42\ (=6\times 7)$ copies of $\mathcal{M}$.

Similar to the case of parametrizing onto an isosceles right-angled triangle, we start with 42 copies of $\mathcal{M}$, $\mathcal{M}_i,\ i=1,...,42$, all embedded identically as $\mathcal{M}$ in $\mathbb{R}^3$. For $i=1,...,42$, we denote by $v^i_j,\ j=0,1,2$ the vertices of $\mathcal{M}_i$ corresponding to the vertices $v_0,v_1,v_2$ of $\mathcal{M}$ respectively. We "glue" the 42 copies according to the partial "gluing instructions" in Figure \ref{fig:hexagons7}. For each $\mathcal{M}_i$ we correspond a triangle in the figure. For each $v_0^i,\ i=1,...,42$, we correspond the corner at the center of the colored hexagon of the triangle attached to $\mathcal{M}_i$. For each $j=1,...,42$, let $v_1^j,v_2^j$ correspond to the two other corners of the triangle attached to $\mathcal{M}_i$, such that $v^i_1$ corresponds to the vertex marked by '1', and $v^i_2$ corresponds to the vertex marked by '2'. 

\begin{figure}[ht!]
\begin{center}
  \includegraphics[trim=5.2cm  7.5cm 4cm 3cm ,clip, scale=0.7]{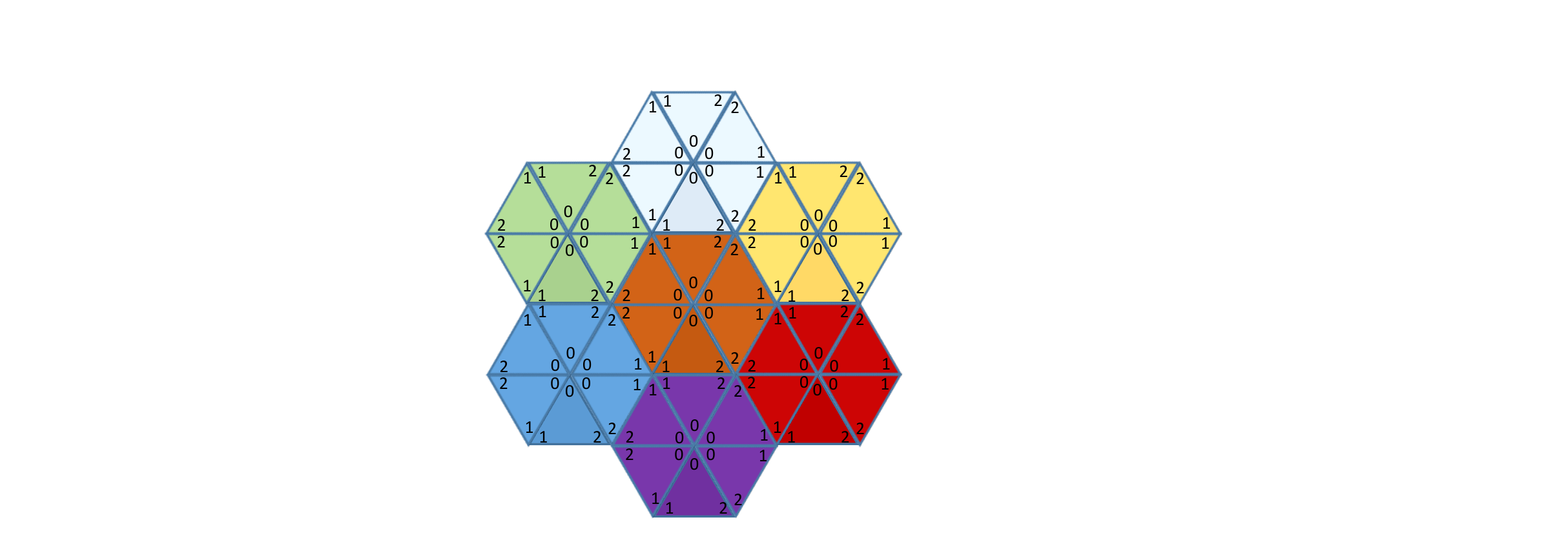}
\caption{Gluing instructions for the copies of the disk mesh comprising the torus.}
\label{fig:hexagons7}
\end{center}
\end{figure}

We now unite the $\mathcal{M}_i,\ i=1,,,.42$, into one connected mesh according to Figure $\ref{fig:hexagons7}$. This is done by looking at a bordering edge between two triangles in the figure. If the triangles correspond to $\mathcal{M}_{i_1}$ and $\mathcal{M}_{i_2}$, the bordering edge corresponds to a path of edges between $v^{i_1}_1$ and $v^{i_1}_2$ in $\mathcal{M}_{i_1}$, and between $v^{i_2}_1$ and $v^{i_2}_2$ in $\mathcal{M}_{i_2}$. Looking at all the edges in Figure \ref{fig:hexagons7} bordering two distinct triangles, we unite all the paths in $\mathcal{M}_i,\ i=1,...,42$ corresponding to these edges, and we obtain one connected simplicial complex denoted by $\mathcal{M}'$.

We look at the boundary of $\mathcal{M}'$. The boundary consists of 21 paths corresponding to the 21 boundary edges of the seven hexagons tiling in Figure \ref{fig:hexagons7}. A boundary edge in Figure \ref{fig:hexagons7} corresponds to a path between vertices $v^{i_0}_1$ and $v^{i_0}_2$ on $\mathcal{M}_{i_0}\subset\mathcal{M}'$ for some $i_0\in\{1,2,...,42\}$. If we were to extend the tiling of the regular triangle along with the marking of the corners, the edge in the figure would be bordering the original triangle and another triangle in the extended tiling. This other triangle is equivalent to one of the tiles in Figure \ref{fig:hexagons7}, say the triangle that corresponds to $\mathcal{M}_{l_0}$. We then unite the path between $v^{i_0}_1$ and $v^{i_0}_2$ in $\mathcal{M}_{i_0}\subset\mathcal{M}'$, with the path between $v^{l_0}_1$ and $v^{l_0}_2$ $\mathcal{M}_{l_0}\subset\mathcal{M}'$. We do this for all the 21 paths corresponding to boundary edges in Figure \ref{fig:hexagons7}. We obtain a simplicial complex $\widetilde{\mathcal{M}}$.

As a surface $\widetilde{\mathcal{M}}$ is homeomorphic to $\mathbb{R}^2/\Lambda_R$, as one can be convinced from considering Figure \ref{fig:hex_tile}, and employing Tutte embedding as before.

Let $S_{R}:\widetilde{\mathcal{M}}\rightarrow \widetilde{\mathcal{M}}$ be the map which "rotates" $\widetilde{\mathcal{M}}$ one step counter-clockwise around the vertex corresponding to the center of the brown hexagon in Figure \ref{fig:hexagons7}: meaning it maps $\mathcal{M}_k\subset\widetilde{\mathcal{M}}$ for some $k\in \{1,..,42\}$ to $\mathcal{M}_{k'}$, where the triangle corresponding to $\mathcal{M}_{k'}$ is the triangle corresponding $\mathcal{M}_{k}$ rotated by 60 degrees counter-clockwise around the center of the brown hexagon. Let $\Phi:\widetilde{\mathcal{M}}\rightarrow \mathbb{R}^2/\Lambda_R$ be the Dirichlet optimal embedding map. Denote by $p_0$ the image under $\Phi$ of the vertex of $\widetilde{\mathcal{M}}$ corresponding to the center of the brown hexagon. Let $H_{RF}:\mathbb{R}^2/\Lambda_R\rightarrow \mathbb{R}^2/\Lambda_R$ be the automorphism of $\mathbb{R}^2/\Lambda_R$, which rotates the torus 60 degrees clockwise around pivot point $p_0$, and reflects each regular triangle along its height emanating from the corner marked by 0 (in the tiling of the torus consisting of 42 regular triangles). We call it the "rotate and flip map".

One has that $\Phi = H_{RF}^{k}\circ\Phi\circ S_R^{k}$ for $k=1,2,3,4,5$. Again, simply by considering the linear equations whose solution is the mapping of the vertices of $\widetilde{\mathcal{M}}$. Similar to before, for $k=1,2,3,4,5$, $\Phi(\widetilde{\mathcal{M}})\cong\Phi\circ S_R^{k}(\widetilde{\mathcal{M}})$, and so we conclude that for $k=1,2,3,4,5$, $\Phi(\widetilde{\mathcal{M}}) = H_{RF}^{k}\circ\Phi(\widetilde{\mathcal{M}})$. Putting it in words, the mapping of the disk in $\widetilde{\mathcal{M}}$ corresponding to the brown hexagon in \ref{fig:hexagons7} has a "rotate and flip symmetry", for a 60 degrees counter-clockwise rotation.

We define a "translate map" $S_T:\widetilde{\mathcal{M}}\rightarrow \widetilde{\mathcal{M}}$. The translate map $S_T$ translates the torus cyclically in the following way: it maps the disk $\widetilde{\mathcal{M}}$ corresponding to the brown hexagon, to the disk corresponding to the light blue hexagon lying to the north of it; it maps the disk corresponding to the purple hexagon to the disk corresponding to the brown hexagon lying north of it; the disk corresponding to the red hexagon to the disk corresponding to the yellow hexagon; and so on. One can refer to \ref{fig:hex_tile} and the marked fundamental domain, to see how the mapping translates the hexagons (yellow goes to dark blue, light blue goes to red and so on). The mapping by $S_T$ of the disks in $\widetilde{\mathcal{M}}$ corresponding to the partition of the hexagon into triangles, is such that the correspondences of the images of the disks to triangles preserves the orientation of the original correspondences. By looking at Figure \ref{fig:hex_tile} one can see that the order of $S_T$ is 7. 
 
We now define the "conjugate" map of $S_T$ on the torus, which we will denote by $H_T$. Let $q\in \mathbb{R}^2/\Lambda_R$ be the vector equal to the difference between the center of the brown hexagon and the center of the purple hexagon, if one has a seven hexagons tiling of the torus as in Figure \ref{fig:hexagons7} where the center of the brown hexagon is at $p_0$ (previously defined). Define $H_T:\mathbb{R}^2/\Lambda_R\rightarrow \mathbb{R}^2/\Lambda_R$ to be the mapping $x\mapsto x+q$  for $x\in\mathbb{R}^2/\Lambda_R$. We have that for $k=1,2,3,4,5,6$, $\Phi = H_T^k\circ\Phi\circ S_T^k$, and so we conclude that if $\mathcal{M}^H\subset\widetilde{\mathcal{M}}$ is the disk corresponding to the brown hexagon in the figure, then $\Phi (\mathcal{M}^H)$ is exactly a hexagon. This is because $\Phi (\widetilde{\mathcal{M}})$ is a union of translates of $\Phi (\mathcal{M}^H)$, where the interiors are disjoint, and if it weren't the case, then a contradiction for the injectivity or the surjectivity of the embedding would follow. Going back to the "rotate and flip map", we conclude that for $i=1,...,42$, each copy of the original mesh, $\mathcal{M}_i$, is embedded by $\Phi$ onto an equilateral triangle. These embeddings are identical up to translation of the range, and thus each has equal minimal Dirichlet energy.

\begin{figure}[h!]
\begin{center}
\includegraphics[trim=0cm  0cm 0cm 0cm ,clip, scale=0.4]{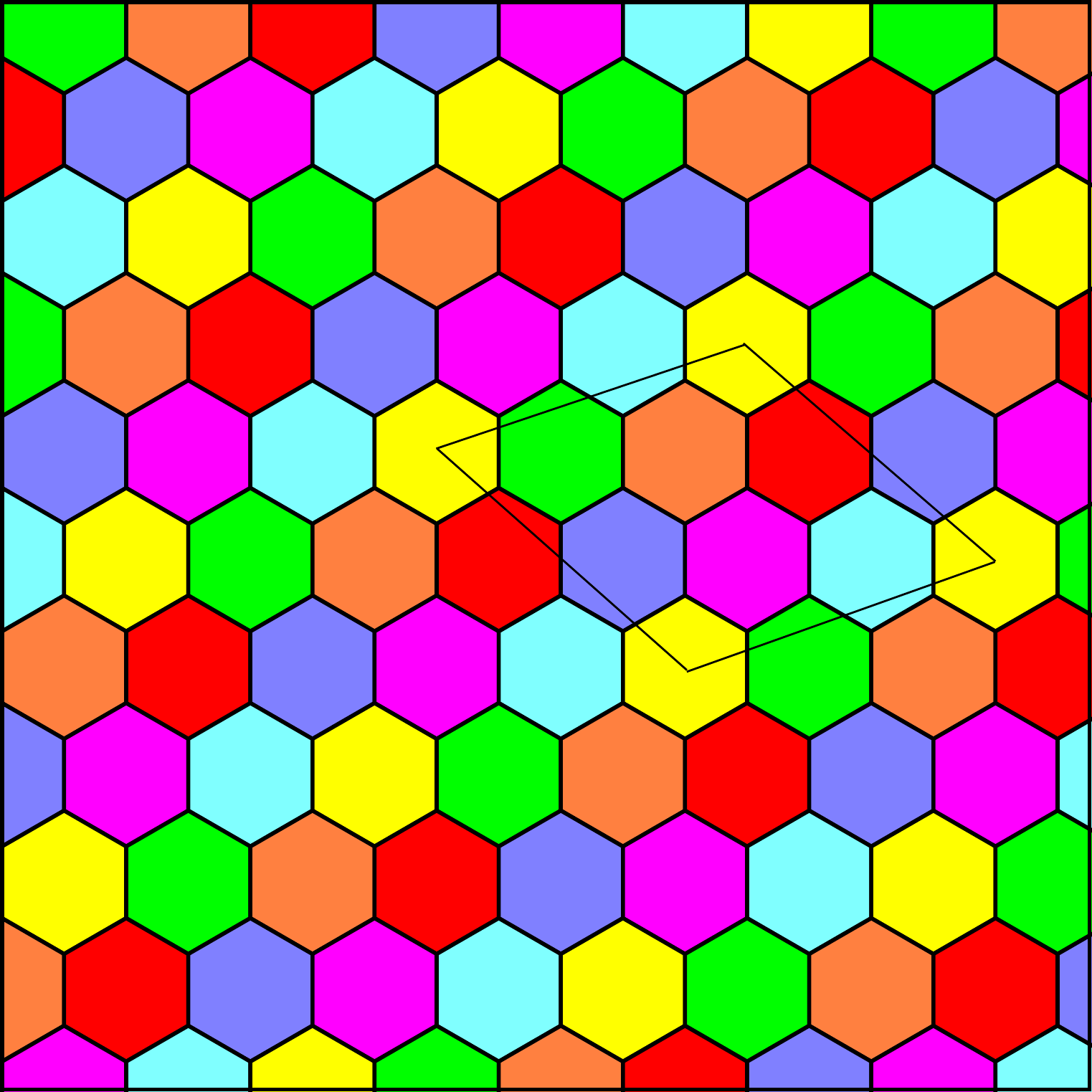}
\caption{Hexagonal tiling of the plane with marked fundamental domain.}
\label{fig:hex_tile}
\end{center}
\end{figure}

\section{The rectangular case}
Given a disk-type mesh $\mathcal{M}$, and four marked vertices on its boundary, we can make four copies of $\mathcal{M}$ and glue them according to Figure \ref{fig:squares} making a torus $\widetilde{\mathcal{M}}$. The embedding of $\widetilde{\mathcal{M}}$ (with "cotan weights") onto $\mathbb{R}^2/\mathbb{Z}^2$ will have horizontal and vertical symmetries, and each copy of $\mathcal{M}$ will be embedded optimally onto a square. The proof is similiar to the previous proofs. One can also replace $\mathbb{Z}^2$ by a different lattice generated by two orthogonal vectors.

\begin{figure}[ht!]
\begin{center}
\includegraphics[trim=0cm  7.5cm 13.5cm 4.7cm ,clip, scale=0.4]{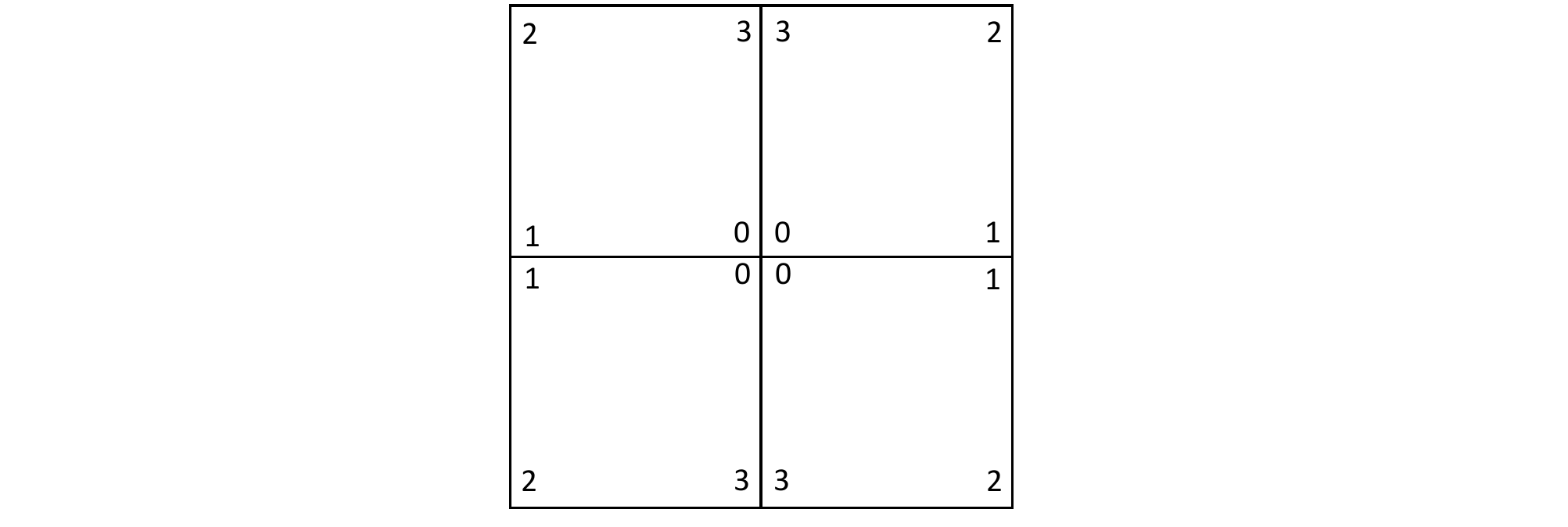}
\caption{Gluing instructions for mapping onto a square.}
\label{fig:squares}
\end{center}
\end{figure}

\section{Complexity of the mappings}
In each of the three embeddings described so far - embedding onto a right-angled isosceles triangle, onto an equilateral triangle and onto a rectangle, we start with disk-type mesh $\mathcal{M}$ with, say, $n$ vertices. For the matter of the proof, we construct a torus out of copies of $\mathcal{M}$. However, if one wishes to implement such an embedding, this could be done with solving two systems of $n$ linear equations with $n$ variables. The variables in one system of equations correspond to the $x$ coordinates of the image of the vertices, and in the other system of equations they corresponds to the $y$ coordinates of the image of the vertices. In each of the cases, we embed the torus, which has more then $n$ vertices, but the proofs show that the embeddings of the different copies are related by affine maps, and so one can actually write an $n$ linear equations systems for an embedding of one copy onto each of the domains.

In the more elaborate case in the next section, computing the embedding requires solving two systems of full rank linear equations, each with approximately $63*n$ variables. This complexity cannot be reduced, and is the reason to expect the method to provide state-of-the-art low distortion parameterizations in applications (for example in the PIEZO1 and PIEZO2 case).

\section{Embedding of a 3-fold symmetric sphere}

In this section we let our simplicial sphere $\widetilde{\mathcal{S}}$ be a simplicial complex which is homeomorphic to $\mathcal{S}^2\subset\mathbb{R}^3$. We have the freedom to choose any three symmetric simple paths, mutually disjoint, on $\widetilde{\mathcal{S}}$ emanating from $p_{\mathcal{O}}$ - one of the points which lie on the sphere and on the axis of symmetry. We now make 63 copies of the sphere $\widetilde{\mathcal{S}}$. On each copy we have the same identical paths as on the original simplicial sphere. For each copy denote the paths on it by $\gamma^i_0,\gamma^i_1,\gamma^i_2$, where $i$ is the index of the copy ($i=1,...,63$). 

We create a torus from the copies, by cutting and stitching along paths $\gamma^i_0,\gamma^i_1,\gamma^i_2, i=1,...,63$. The gluing instructions can be found in Figure \ref{fig:hex_spheres}. Each hexagon not lying on the two gray rectangles, represents one of the simplicial sphere. In Figure \ref{fig:hex_spheres}, '11' and '12' denote the two sides created by "cutting" along path $\gamma^i_0$ (considering sphere $i$). '21' and '22' denote the two sides created by "cutting" along path $\gamma^i_1$, and '31' and '32' denote the two sides created by "cutting" along path $\gamma^i_2$.

Looking at the hexagons copied to the top gray rectangle from the bottom part, and the hexagons copied to the right gray rectangle from the left, one can be convinced that the stitching actully yields a torus.

We then embed the torus $\widetilde{\mathcal{T}}$ onto $\mathbb{R}^2/\Lambda_R$ ($3\times 3 \times 7=63$, recall that $\Lambda_R$ is the rhombus lattice define before). By analogous arguments $\widetilde{\mathcal{S}}\subset\widetilde{\mathcal{T}}$ (that is, any sphere in the constructed torus) is mapped to a tile which has 3-fold rotational symmetry. The map has optimal Dirichlet energy under the symmetry and tiling constraint.

The conjugate maps which prove the symmetric nature of the embedding are the rotations of all spheres by $2\pi/3$ around their axis of symmetry and the rotations of the plane by $2\pi/3$ around the image of the antipodal point of $p_{\mathcal{O}}$. Then we consider conjugate pairs of maps - a translation map on the torus and a map that maps one sphere to another accordingly. The proof is similar to the proof in the embedding of a disk onto an equilateral triangle.

In Figure \ref{fig:embedding_noedge_full} the embedding of the whole torus can be seen. In Figure \ref{fig:two_spheres} the original sphere that we took can be seen. The torus was constructed out of 63 copies of the depicted sphere.

\begin{figure}[ht]
\begin{center}
\includegraphics[trim=1.5cm  1.5cm 0cm 2.5cm ,clip, scale=0.6]{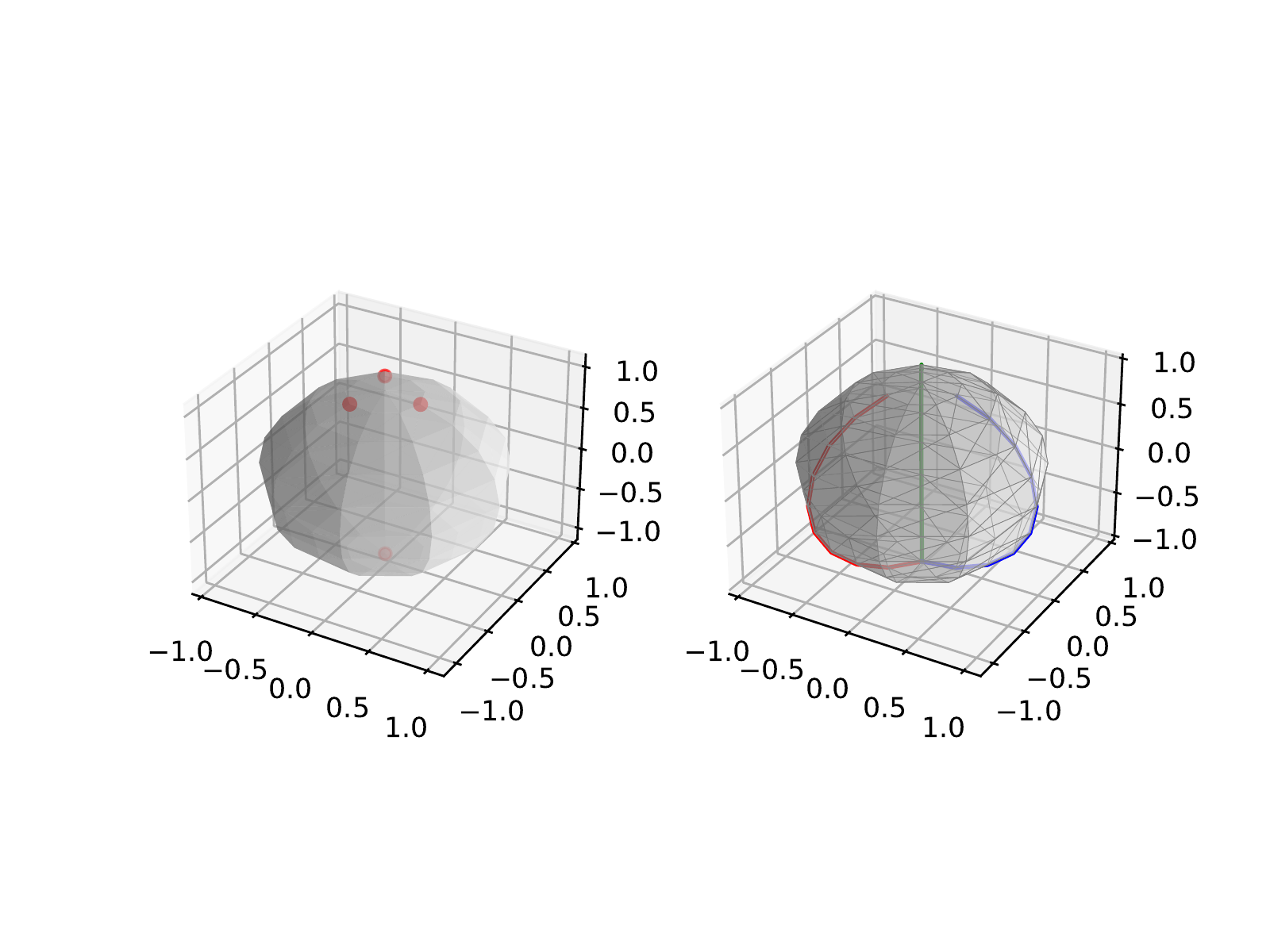}
\caption{The original spheres - on the left with the special points from which the cuts run, on the right with three cuts in different colors.}
\label{fig:two_spheres}
\end{center}
\end{figure}

\begin{figure}[ht!]
\begin{center}
\includegraphics[trim=0cm  20cm 3cm 0cm ,clip, scale=0.2]{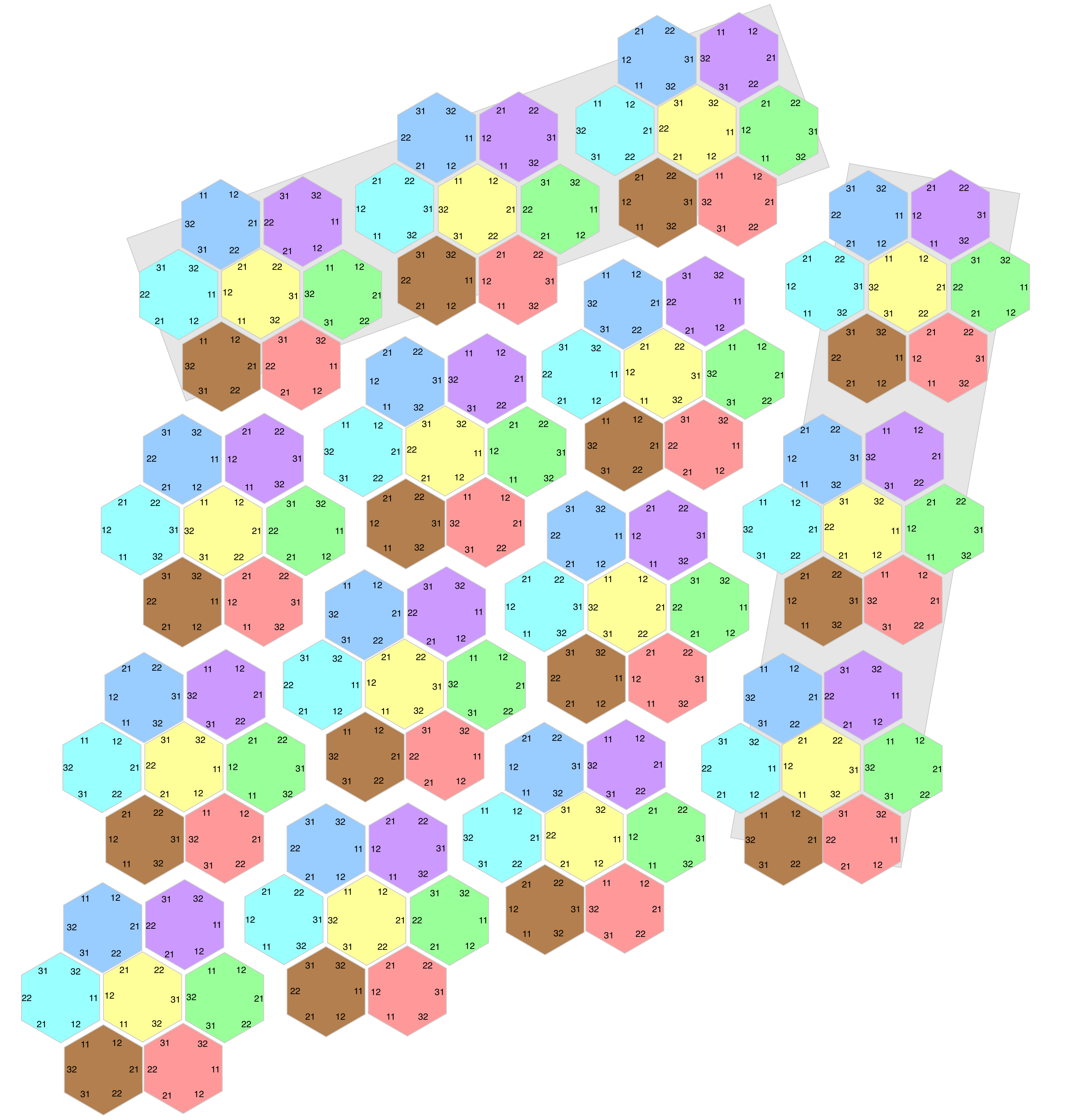}
\caption{Gluing instructions for stitching the 63 copies of the sphere-type mesh.}
\label{fig:hex_spheres}
\end{center}
\end{figure}

\section{Conclusion}
We proved the validity of a method to parametrize (embed) in the plane a 3-fold rotationally symmetric sphere-type mesh. The embedding is preceded by a novel construction of a branched covering of the sphere with a torus domain, which is composed of 63 copies of the spheres, positioned at the same place in the space. For discrete torus surfaces (genus 1 meshes), one can write two systems of weighted harmonic equations for the x and y coordinates of the images of the vertices of the meshes. The weights can be such that the mapping will be conformal (it minimizes the LSCM energy, a discrete notion of "conformal"). We demonstrated the embedding of a discrete sphere using this method (with equal weights for simplicity). The equations to be solved were full rank after fixing one image point of the embedding, which fits the analysis that this branched covering is the minimal one possible, with respect to the number of identically embedded 2-spheres, that allows a branched covering with the same ramification structure for each branch point. The implementation used topological algorithms on graphs needed to locate the first homology group generators of the torus constructed and cutting it along these loops (see \cite{Jin2018}). We believe that the novel method presented in this paper will prove to be highly useful in applications, due to its linearity and optimality.


\printbibliography
\Addresses

\end{document}